\documentclass[11pt]{article}
\usepackage{amsmath,amssymb} % Math packages
\usepackage{color,soul}
\usepackage{multirow}
\usepackage{url,setspace, color, colortbl}
\usepackage{bm,amsmath,amsfonts}
\usepackage[round]{natbib}
\usepackage[utf8]{inputenc}
\usepackage[english]{babel}
\usepackage{booktabs}
\usepackage{multirow}
\usepackage{siunitx}
\usepackage{amsthm}
\usepackage{color,hyperref,hypernat}
\definecolor{darkblue}{rgb}{0.0,0.0,0.3}
\hypersetup{colorlinks,breaklinks,
            linkcolor=darkblue,urlcolor=darkblue,
            anchorcolor=darkblue,citecolor=darkblue}
\usepackage{mathtools}

\usepackage{graphicx,psfrag,epsf}
\usepackage{enumerate, upgreek}
\usepackage{mathtools}
\usepackage{mathrsfs}
\definecolor{Gray}{gray}{0.9}
\usepackage[utf8]{inputenc}
\usepackage{algorithm}
\usepackage[noend]{algpseudocode}
\algnewcommand\algorithmicinput{\textbf{Update:}}
\algnewcommand\Update{\item[\algorithmicinput]}
\algnewcommand\Initialize{\item[{\textbf{Initialize:}}]}
\algnewcommand\Input{\item[{\textbf{Input:}}]}
\algnewcommand\Output{\item[{\textbf{Output:}}]}
\algnewcommand\Set{\item[{\textbf{Set:}}]}
\algnewcommand\Define{\item[{\textbf{Define:}}]}
\algnewcommand\Fix{\item[{\textbf{Fix:}}]}
\algnewcommand\Transform{\item[{\textbf{Transform:}}]}
\usepackage[nottoc]{tocbibind}
\usepackage{booktabs, array}
\usepackage{longtable,rotating}
\usepackage{nomencl}
\usepackage{subfigure}
\usepackage{fmtcount}
\usepackage{fancyhdr}
\usepackage[small,bf]{caption}

\usepackage{amssymb}
\newtheorem{theorem}{Theorem}
\newtheorem{lemma}{Lemma}
\newtheorem{proposition}{Proposition}

\usepackage{graphicx}
\usepackage{enumitem}
\usepackage{authblk}

\title{\textbf{Statistical inference with normal-compound gamma priors in regression models}}
\author[1]{Ahmed Alhamzawi\thanks{ahmed.alhamzawi@qu.edu.iq}}
\author[2]{Gorgees Shaheed Mohammad\thanks{gorgees.alsalamy@qu.edu.iq}}
\affil[1]{University of AL-Qadisiyah, College of Science, Department of Mathematics, Iraq}
\affil[2]{University of AL-Qadisiyah, College of Education, Department of Mathematics, Iraq}

\begin{document}

\maketitle

\begin{abstract}
\normalsize
Scale-mixture shrinkage priors have recently been shown to possess robust empirical performance and excellent theoretical properties such as model selection consistency and (near) minimax posterior contraction rates. In this paper, the normal-compound gamma prior (NCG) resulting from compounding on the respective inverse-scale parameters with gamma distribution is used as a prior for the scale parameter. Attractiveness of this model becomes apparent due to its relationship to various useful models. The tuning of the hyperparameters gives the same shrinkage properties exhibited by some other models. Using different sets of conditions, the posterior is shown to be both strongly consistent and have nearly-optimal contraction rates depending on the set of assumptions. Furthermore, the Monte Carlo Markov Chain (MCMC) and Variational Bayes algorithms are derived, then a method is proposed for updating the hyperparameters and is incorporated into the MCMC and Variational Bayes algorithms. Finally, empirical evidence of the attractiveness of this model is demonstrated using both real and simulated data, to compare the predicted results with previous models.
\end{abstract}

\section{Introduction}
Normal linear regression models are likely the most appropriate   statistical models
used to  illustrate  the influence of a set of covariates  on an outcome variable. The normal linear model can be written as:
\begin{eqnarray}\label{ahmed_1}
\boldsymbol{y} = X\boldsymbol\beta + \boldsymbol\epsilon, \,
\end{eqnarray}
where $\boldsymbol{y} = (y_1, \cdots, y_n)^{T}$ is a vector of observed values, $X=(\boldsymbol{x}_1, \cdots, \boldsymbol{x}_p)$ is an $n \times p$ design matrix of covariates with $\boldsymbol{x}_j=(x_{1j}, \cdots, x_{nj} )^T$, $\boldsymbol{\beta}=(\beta_1, \cdots, \beta_p)^T $  is a $p \times 1$ of unknown regression coefficients, $\boldsymbol\epsilon=(\epsilon_1, \cdots, \epsilon_n)^T$ and $\epsilon_i\sim N(0, \sigma^2)$  where $\sigma^2$ is the unknown variance.
Under  model (\ref{ahmed_1}), it is assumed that only a subset of the covariates are active in the regresssion, so that the covariate  selection problem is to  recognize this unknown subset of covariates. Various  methods have been developed over the years for identifying the active covariates in the regression. For example see, the traditional model selection methods such as,
Akaike information criterion \cite[AIC;][]{Akaike_1973}, Bayesian information criterion  \cite[BIC;][]{schwarz1978estimating}, Mallows' C$_p$-statistic \citep{mallows1973some} and the deviance information criteria  \citep[DIC;][]{spiegelhalter2002bayesian}. These  criterions were used to compare $2^p$ candidate models. However, when  $p$ is large, it needs highly greedy computations.

Variable selection by penalized least squares regression has received considerable attention in the last three decades, for example see, the Lasso  \citep{Tibshirani_1996}, the adaptive Lasso  \citep{zou2006adaptive}, the elastic net \citep{Zou_and_Hastie_2005}, the adaptive elastic net \citep{Zou_and_Zhang_2009},   the group Lasso \citep{Yuan_and_Lin_2005}, the bridge \citep{frank1993statistical}, the group bridge \citep{huang2009group} and the reciprocal Lasso \citep{song2015high}. These approaches  provide theoretically attractive estimators as well as it     provide an efficacious and computationally appealing alternative to the  classical model selection approaches. However, despite being theoretically attractive   and have nice properties in terms of variable selection and coefficient  estimation, these approaches  usually cannot provide  valid standard errors \citep{kyung2010penalized}.

Bayesian inference  approaches overcome this problem by giving a more valid measure of the standard errors based on a stationary geometrically ergodic MCMC algorithms \citep{kyung2010penalized}. Therefore, in the last two decades, great works have been done in the direction of Bayesian methodology (see, the Bayesian Lasso \citep{Park_and_Casella}, the Bayesian adaptive Lasso \citep{alhamzawi2018bayesian}, the Bayesian elastic net \citep{li2010bayesian1},  the Bayesian adaptive elastic net \citep{gefang2014bayesian}, the Bayesian group Lasso \citep{xu2015bayesian} and the bayesian bridge \citep{polson2014bayesian}). These Bayesian approaches  can be  acquired by putting  a suitable prior distribution on the regression coefficients that will mimic the property of the corresponding   penalty. The performance of resulting estimators  depends on the form of the priors for $\bm{\beta}$. Most of the above Bayesian approaches  are based on the scale mixture of normals (SMN)  of the associated prior distribution for the corresponding penalty. Very recently, when $p\gg n$,  much Bayesian work has been done on sparse regression using   spike-and-slab priors with point masses at zero \citep{castillo2015bayesian,rovckova2018spike} and  nonlocal priors \citep{johnson2012bayesian,shin2018scalable,mallick2021reciprocal}.

In this paper, our main goal to study the properties of the Normal-Compound gamma model and show its robustness at handling
sparsity and non-sparsity signals and it show it is relationship to different models as special cases. In section 2, we present our model, which is based on a multivariate-normal scale mixture can be viewed as a generalization to a large array of useful models by connecting to a large family of hierarchical priors as special cases such as the horseshoe (half-cauchy) prior, beta prime prior, generalized beta prior, etc. The elegance of these models promotes us to study their generalizations and compare the results with for different choices of hyperparameters. In section 3, we derive the Gibbs sampler and the variational Bayes and then incorporate the relevant empirical Bayes method. Additinally, in section 4 we show that this prior achieves posterior consistency for both $p \leq n$ and $p\geq n$ as $n\rightarrow 0$. Finally, in section 5 we use real and simulated data to test the performance of these approaches.

\newpage

\section{ Scale Mixture of Compound Gamma Distribution }
\textbf{Proposition}

A compound gamma density resulting from the compounding of $N$ gamma distributions  can be written as
\begin{equation}\label{marginal-prior1}
\pi(\textbf{x}) = \int_0^\infty\ldots\int_0^\infty
\left[\prod_{i=1}^{N}\frac{\textbf{z}_{i+1}^{c_{i}}}{\Gamma(c_i)}\textbf{z}_i^{c_{i}-1}
\exp\left\{-\textbf{z}_{i}\textbf{z}_{i+1}\right\}\right]
d\textbf{z}_2 \ldots d\textbf{z}_N
\end{equation}
where $\textbf{z}_1=\textbf{x}$, $\textbf{z}_{i}=(z_{i1},\dots,z_{ip})$ and $\textbf{z}_{N+1}=\phi$ is some constant that may be determined from the data.
The proof of this result is straightforward and thus will be omitted. This density can be considered as generalization of different previous proposals with
different behaviors at their respective tails and origin depending on the value of $N$ and $c_{i}$. Some popular special cases include
the three-parameter Beta Distribution \citep{armagan2011generalized}, the scaled Beta2  family of distributions \citep{perez2017scaled} when $N=2$, the Beta prime distribution when $N=2$ and $\phi=1$ \citep{bai2018beta} and
the horseshoe prior for $N=4$ and $c_1=c_2=c_3=c_4=1/2$ \citep{carvalho2010horseshoe},
where $\beta_i|\text{rest} \sim \mathcal{N}(0,\sigma^2 {z}_1),\quad {z}_1^{1/2} \sim \text{C}^{+}(0,{z}_2), \ \text{and} \
{z}_2^{1/2} \sim \text{C}^{+}(0,1)$. The parameter $\beta$ maybe estimated using the scale-mixture method.
The normal-compound gamma scale mixture model is given by

\begin{align}\label{ahmed2}
&\beta_j|\sigma^2,z_1 \sim \mathcal{N}(0,\sigma^2 {z}_1)\\
&{z}_1 \sim \mathcal{CG}(c_1,\ldots,c_N,\phi)\nonumber\\
&\sigma^2 \sim \mathcal{IG}(c_{0},d_{0})\nonumber
\end{align}

To simplify the complexity of the prior proposed in (\ref{marginal-prior1}), we propose an alterative way to compute the full conditionals by using the following equivalence

\begin{proposition}\label{equivalence-proposition}
If ${z}_1 \sim \mathcal{CG}(c_1,\ldots,c_N,\phi)$, then
\begin{flalign}\label{qprior}
\text{(1)}\,\, {z}_1 \sim \mathcal{G}(c_1,{z}_2),{z}_2 \sim \mathcal{G}(c_2,{z}_3),\ldots,{z}_{N} \sim \mathcal{G}(c_{N},\phi)&&
\end{flalign}
\begin{flalign}
\text{(2)}\,\, {z}_1 \sim \mathcal{G}(c_1,1),{z}_2 \sim \mathcal{IG}(c_2,1),\ldots,{z}_{N} \sim \mathcal{AG}(N,c_N,\phi)&&
\end{flalign}
where $\mathcal{AG}(N,a,b)=\begin{cases}
      \mathcal{G}(a,b) & \text{odd}\,\,\,{N} \\
      \mathcal{IG}(a,b) & \text{even}\,\,\,{N}
   \end{cases}$
 and $\mathcal{G}(a,b)$ is the gamma distribution with shape $a$ and inverse scale
(rate) parameter $b$.
\end{proposition}

\begin{proof}
The proof of this equivalence is provided in the Appendix.
\end{proof}

\begin{figure}[h]
\begin{center}
    \includegraphics[width=80mm]{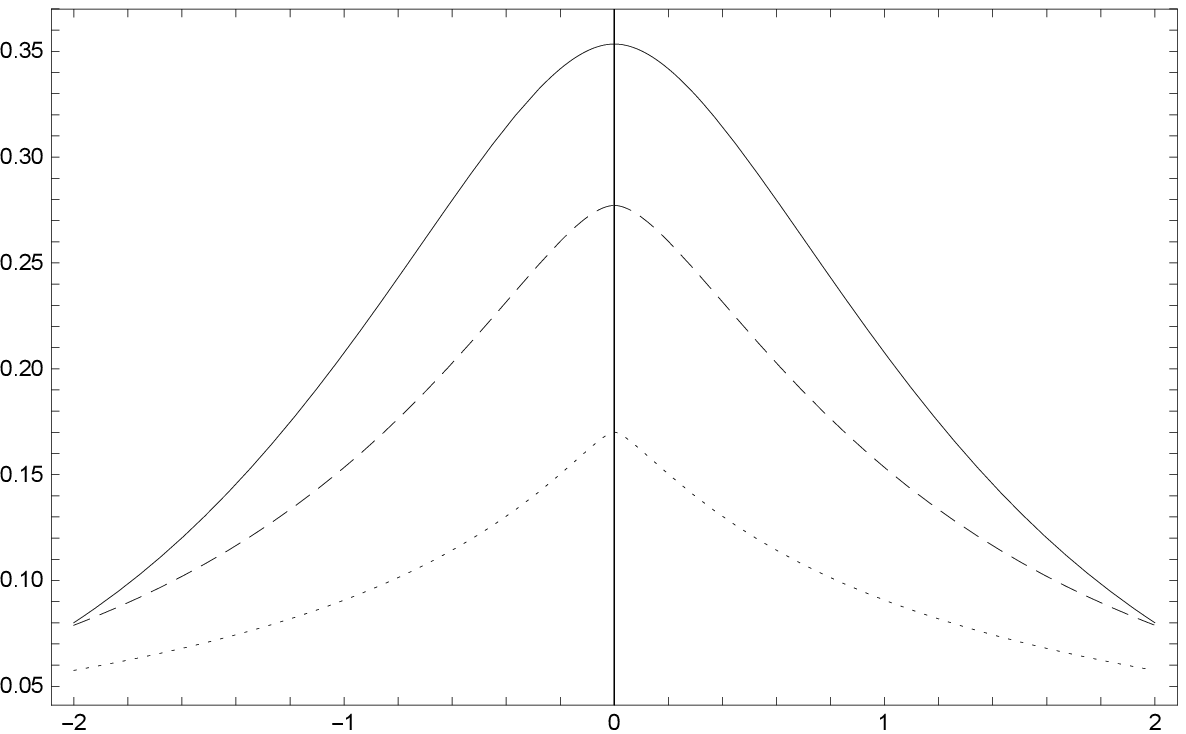}
    \includegraphics[width=80mm]{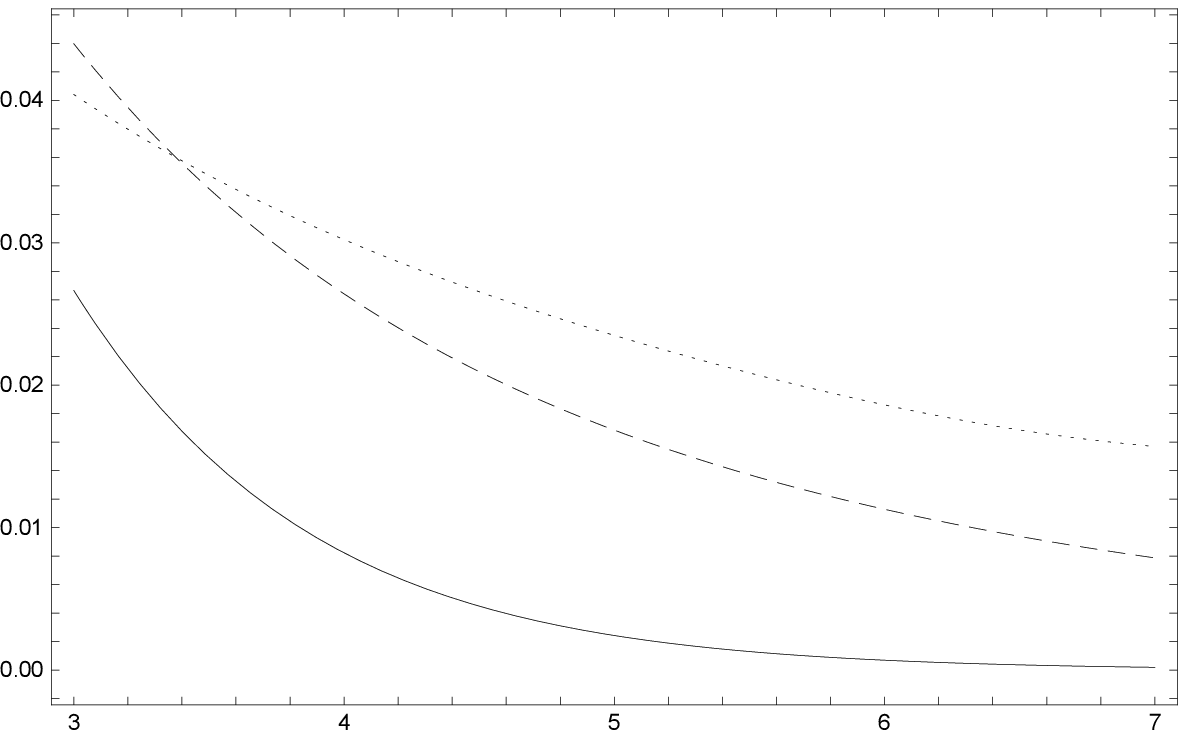}
    \includegraphics[width=80mm]{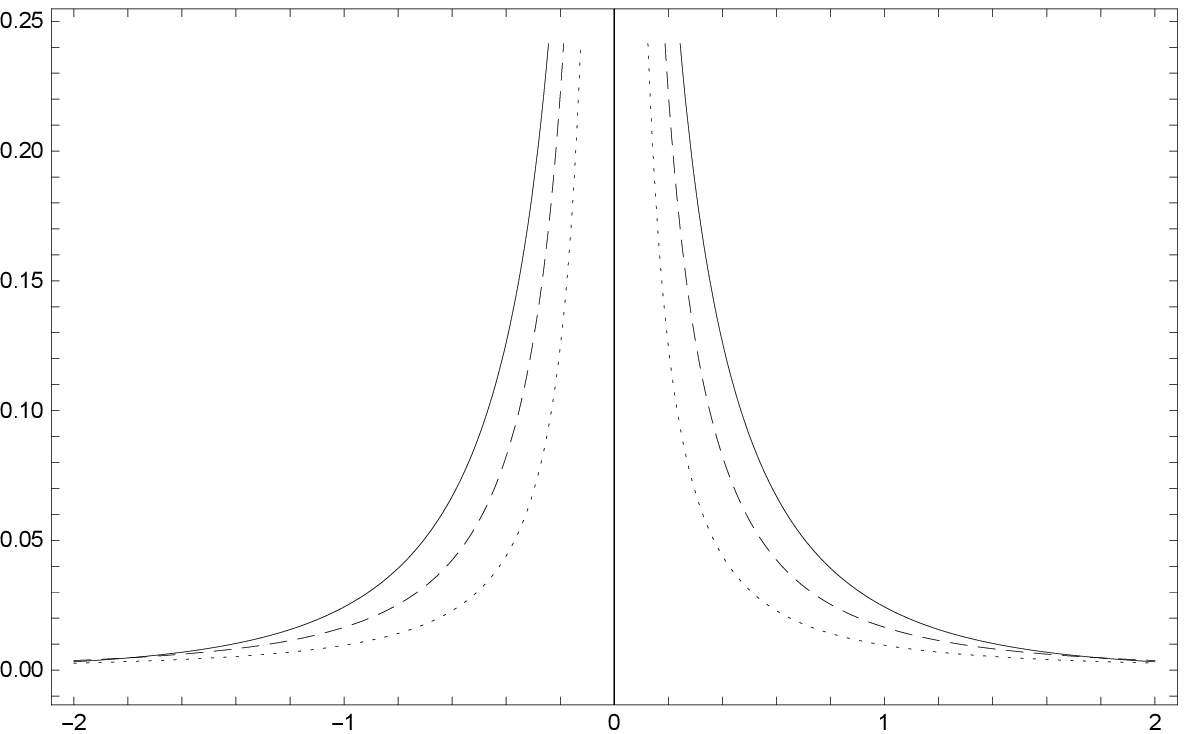}
    \includegraphics[width=80mm]{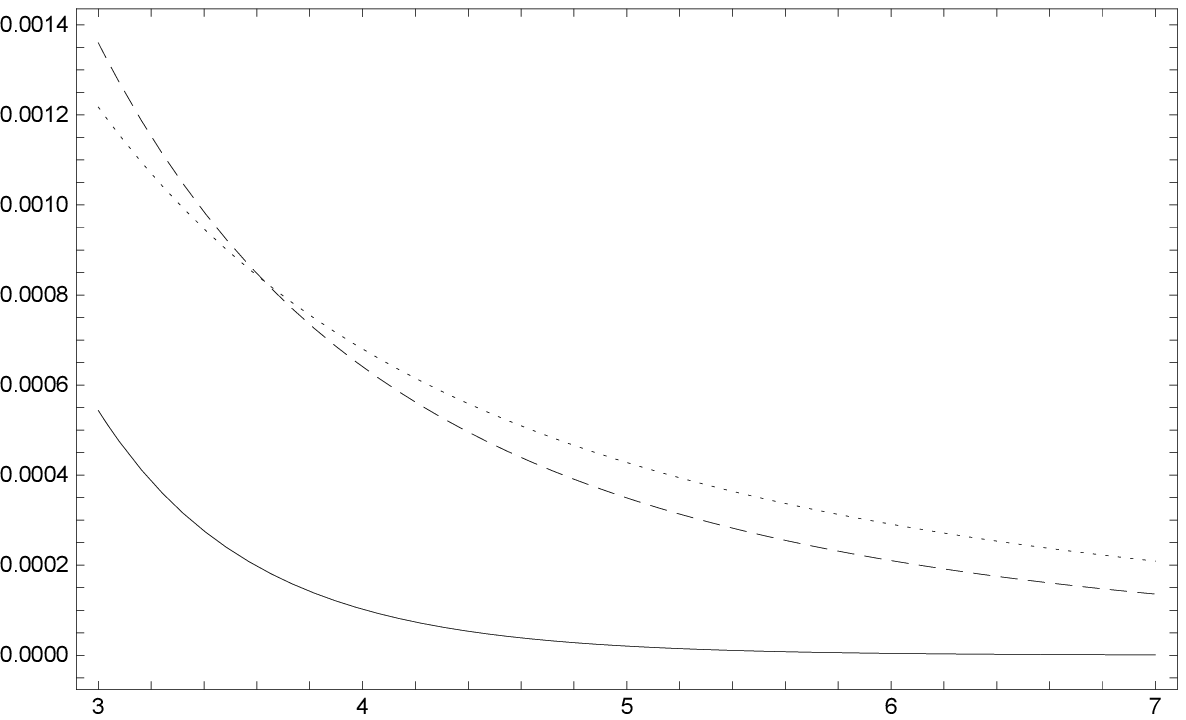}
    \caption{Plot of the marginal prior for different values of $c_1$ (top $c_1=1.5$, bottom $c_1=0.15$) and $N$. Solid line for $N=2$,
     dashed line for $N=4$ and dotted line for  $N=8$}
    \label{figure for different N and a}
\end{center}
\end{figure}

The properties of our prior are indicated in Figure \ref{figure for different N and a} for different values of $N$ and $c_1$.
We can clearly see that for large values of $c_1$, as the value of $N$ increase we have a falter head with a heavier tail.
Conversely, for smaller $c_1$, we see that we have a pole at zero and as $N$ increases the pole diverges faster which   indicates a good sparse prior.
This clearly demonstrates how our model serves as candidate for both sparse and non-sparse model.

To continue with our study, it is necessary to use some notations: $a_n\prec b_n$ or $a_n=o(b_n)$ will donate $\lim a_n/b_n=0$,
while $a_n\lesssim b_n$ or $a_n=O(b_n)$ means $a_n<C b_n$ for some $C>0$ and
we write $a_n\asymp b_n$ if $a_n\lesssim b_n$ and $b_n\lesssim a_n$,

\section{Posterior Inference and Empirical Bayes}

Using the hierarchical representation in (\ref{ahmed2}) we can derive the full conditional probability as following:

\begin{itemize}

\item{Update $\beta$}

\begin{align}
\begin{split}
P(\beta|X,y,...) & \propto P(y|\beta,...)\pi(\beta),\\
&\propto\exp\left\{-\frac{1}{2\sigma^2}(y-X\beta)^T (y-X\beta)
-\frac{1}{2\sigma^2}\beta^T \textbf{Z}^{-1} \beta\right\},\\
&=\exp\left\{-\frac{\Sigma}{2\sigma^2}\left(-2{\mu}^T_{\beta}
\beta+\beta^T  \beta\right)\right\},\\
\end{split}
\end{align}

where $\mu_{\beta}= \Sigma^{-1}X^T y$ and $\Sigma=  X^T X+ \textbf{Z}^{-1}$. Therefore, we have the normal distribution $\mathcal{N}(\mu_{\beta},\Sigma^{-1}\sigma^2)$ as a conditional distribution for $\bm{\beta}$. For odd $k$ we have

\item{Update odd ${z}_{k}$}

\begin{align}
\begin{split}
P({z}_{k}|X,y,...) & \propto \pi(\beta_i|{z}_{1},{z}_{2},\ldots,{z}_{N},\sigma^2)\pi({z}_{k})\\
& \propto \frac{1}{\sqrt{{z}_{k}}} \exp\left\{-\frac{\beta^T \textbf{Z}^{-1} \beta}{2\sigma^2}\right\}
\times \left({z}_{k}\right)^{c_k-1} \exp\left\{-{z}_{k}{\varphi_k}\right\},\\
& \propto \left({z}_k\right)^{\left(c_k-\frac{1}{2}\right)-1} \exp\left\{-\frac{1}{2}\left[
\frac{\beta^T \textbf{Z}^{-1}_{-k} \beta}{\sigma^2}\left({z}_{k}\right)^{-1}+2{z}_{k}{\varphi_k}\right]\right\},\\
\end{split}
\end{align}

where $\varphi_k=1+(\phi-1)I(k-N)$. Thus, we have the generalized inverse gaussian distribution
$\mathcal{GIG}\left(\frac{\beta^T \textbf{Z}^{-1}_{-k} \beta}{\sigma^2},2{\varphi_k},c_k-\frac{1}{2}\right)$. For even $k$, we have

\item{Update even ${z}_{k}$}

\begin{align}
\begin{split}
P({z}_{k}|X,y,...) & \propto \pi(\beta_i|{z}_{1},{z}_{2},\ldots,{z}_{N},\sigma^2)\pi({z}_{k})\\
& \propto \frac{1}{\sqrt{{z}_{k}}} \exp\left\{-\frac{\beta^T \textbf{Z}^{-1} \beta}{2\sigma^2}\right\}
\times \left({z}_{k}\right)^{-c_k-1} \exp\left\{-\frac{{\varphi_k}}{{z}_{k}}\right\},\\
& \propto \left({z}_{k}\right)^{-\left(c_k+\frac{1}{2}\right)-1} \exp\left\{-\left[
\frac{\beta^T \textbf{Z}^{-1}_{-k} \beta}{2\sigma^2}+{\varphi_k}\right]\left({z}_{k}\right)^{-1}\right\},\\
\end{split}
\end{align}

which is the inverse-gamma distribution
$\mathcal{IG}\left(c_k+\frac{1}{2},\frac{\beta^T \textbf{Z}^{-1}_{-k} \beta}{2\sigma^2}+{\varphi_k}\right)$.

\item{Update $\sigma^2$}

\begin{align}
\begin{split}
P(\sigma^2|X,y,...) & \propto P(y|\beta,...)\pi(\beta_i|{z}_{1},{z}_{2},\ldots,{z}_{N},\sigma^2)\pi(\sigma^2)\\
&\propto  \left(\sigma^2\right)^{-n/2} \exp\left\{-\frac{(y-X\beta)^T (y-X\beta)}{2\sigma^2}\right\}
\times \left(\sigma^2\right)^{-p/2}    \exp\left\{-\frac{\beta^T \textbf{Z}^{-1} \beta}{2\sigma^2}\right\}\\
&\qquad\qquad\times \left(\sigma^2\right)^{-c_0-1}   \exp\left\{-\frac{d_0}{\sigma^2}\right\}\\
&\propto \left(\sigma^2\right)^{-\left(\frac{n+p+2c_0}{2}\right)-1}
\exp\left\{-\frac{ (y-X\beta)^T (y-X\beta)+\beta^T \textbf{Z}^{-1} \beta+2d_0
}{2\sigma^2}\right\}\\
\end{split}
\end{align}
which is the inverse-gamma distribution $\mathcal{IG}\left(\frac{n+p+2c_0}{2},
\frac{(y-X\beta)^T (y-X\beta)+\beta^T \textbf{Z}^{-1} \beta+2d_0}{2}\right)$.

\end{itemize}

Alternatively, we may use the variational Bayes method \citep{bishop2000variational,jordan1999introduction} with
$q_0^*(\boldsymbol{\beta})\sim\mathcal{N}(\mu^*,V^*)$, $q_k^*({\textbf{z}}_{\text{odd}})\sim\mathcal{GIG}\left(c_k^*,2,c_k - 1/2 \right)$, $q_k^*({\textbf{z}}_{\text{even}})\sim\mathcal{IG}\left(c_k+1/2,c_k^*\right)$ and $q_{N+1}^*(\sigma^2)\sim\mathcal{IG}\left((n+p+2c_0)/2,d_0^*\right)$, where $\mu^*=(X^T X+{\textbf{Z}^*}^{-1})^{-1} X^T y$,
${V^*}=\mathbb{E}_{q_{N+1}^*}(\sigma^2) (X^T X+{\textbf{Z}^*}^{-1})^{-1} $,
${\textbf{Z}^*}^{-1}= \mathrm{diag}( \prod_{k=1}^{N}\mathbb{E}_{q_{k}^*}(z_{k1}^{-1}),\ldots,\prod_{k=1}^{N}\mathbb{E}_{q_{k}^*}(z_{kp}^{-1}))$, $\mathbb{E}_{q_{0}^*}(\beta^2_i)={\mu^*_i}^2+V^*_{ii}$, $\mathbb{E}_{q_{0}^*}\left(\parallel y-X\beta \parallel_2^2\right)=\parallel y-X\mu^* \parallel_2^2 +\text{tr}(X^T X V^*)$, $\mathbb{E}_{q_{}^*}\left(X^T \Lambda X\right)=\sum^p_{i=1}{\beta^*_i}^2 \prod_{k=1}^N \mathbb{E}_{q_{k}^*}(z_{ki}^{-1}) + \text{tr}(\Lambda V^*_{ii})$, $c_{\text{odd}\,k}^*=\mathbb{E}_{q_{0}^*}(\beta^2)\mathbb{E}_{q_{N+1}^*}(\sigma^{-2})\prod_{i=1,i\neq k}^{N}\mathbb{E}(z_i^{-1})$, $c_{\text{even}\,k}^*=\frac{1}{2}\mathbb{E}_{q_{0}^*}(\beta^2)\mathbb{E}_{q_{N+1}^*}(\sigma^{-2})\prod_{i=1,i\neq k}^{N}\mathbb{E}(z_i^{-1})+1$ and $d_0^*=\frac{\mathbb{E}_{q_{0}^*}\left(\parallel y-X\beta \parallel_2^2\right)+\mathbb{E}_{q_{}^*}\left(\beta^T \Lambda^* \beta\right)+2d}{2}$. Then, we try to optimize the evidence lower boundary (ELBO)

\begin{equation}\label{varitational-L}
\mathcal{L}=\mathbb{E}_q \log f(u,\beta,\textbf{z}_1,\ldots,\textbf{z}_N,\sigma^2) - \mathbb{E}_q \log q(\beta) - \sum^N_{k=1}
\mathbb{E}_q \log q(\textbf{z}_k) - \mathbb{E}_q \log q(\sigma^2).
\end{equation}

To evaluate the values of $c_k$, we incorporate an Expectation Maximization (EM) algorithm to the Gibbs sampler by developing a Monte carlo (MCEM) method \citep{wei1990monte} using the expectation of the log complete-data likelihood

\begin{align}
\begin{split}
Q\left(\theta,\theta^\text{old}\right)
=&\,\,\, \sum_{k=1}^{N}\sum_{i=1}^{p} (-1)^{k+1} c_{k} \mathbb{E}_{c_{k}^\text{old}} \left[ \log({z}_{ki})|y
\right]
+c_{N}\log(\phi)
- \sum_{k=1}^{N} \log(\Gamma(c_{k}))+\mathfrak{C}\\
\end{split}
\end{align}

where $\mathfrak{C}$ donates all the terms not containing $c_1,c_2,\ldots,c_{N}$.

\begin{equation}
\Gamma^\prime(c_{k})=\sum_{i=1}^{p}(-1)^{k+1}  \mathbb{E}_{c_{k}^\text{old}} \left[ \log({z}_{ki})|y
\right]
+c_{N}\log(\phi) I(k=N)
\end{equation}

Similarly we can incorporate the EM algorithm to the variational inference method using the Mean Field Variational Bayes (MFVB) with

\begin{equation}
\mathbb{E}[\log(\textbf{z}_{\text{odd}})]=\frac{1}{2}\log\left(\frac{2}{c_k^*}\right)+\log\left(K_{c_k-1/2}^\prime(\sqrt{2c_k^*})\right)
\end{equation}
\begin{equation}
\mathbb{E}[\log(\textbf{z}_{\text{even}})]=\log(c_k^*)-\psi(c_k+\frac{1}{2})
\end{equation}

in (\ref{varitational-L}),where $K^\prime_x(\cdot)$ donates the derivative of the  bessel function of the second kind with respect to $x$.

\section{Model Consistency}

Suppose that $\beta_n^0$ is the true parameter of a high dimensional model with some non-zero components and let $\epsilon\sim N(0,\sigma^2 I)$, where $\sigma^2>0$ is known. To study high dimensionality in the context of posterior consistency, we will assume that as $n\rightarrow\infty$, then $p_n\rightarrow\infty$ and

\begin{enumerate}[leftmargin=1cm,label= (A\arabic*):]
\item   All the covariates are uniformly bounded by 1.
\item   $p_n\gg n$. (High dimensionality)
\item   $\exists \hat{p}(n,p_n)\in \mathbb{N}\ni \hat{p}(n,p_n) \succ s$.
        \newline
        $\exists \lambda_0 \ni \lambda_\text{min} (X^T_\xi X_\xi)\geq n\lambda_0, \forall |\xi|\leq \hat{p}(n,p_n)$,
        where $\xi \subset \{1,\ldots,p_n\}$.
\item   $s_n=o(n/ \log p_n)$.
\item   $\text{max}_j \{\beta_{0j}/\sigma_0\}\leq \gamma_3  E_n $ for some $\gamma_3 \in (0, 1)$, and a
    nondecreasing $E_n$ with respect to $n$.
\end{enumerate}

where  $\lambda_\text{min} (X^T_\xi X_\xi)$ denotes the minimum eigenvalue of $X^T_\xi X_\xi$. For simplicity, we will set $\sigma^2=1$. To demonstrate posterior consistency, we can derive the near-optimal contraction rates
by utilizing the result in \citep{song2017nearly} by the theorem

\begin{theorem}\citep{song2017nearly}\label{pythagorean}
If the Assumptions $[A1-A5]$ are satisfied and further assume a normal-compound gamma prior

 \begin{equation}\label{marginal-prior}
f(\beta) =
\int_0^\infty\ldots\int_0^\infty
\left[
\frac{\exp\left\{-{\beta^2}/{\left(2\sigma^2 z_1\right)}\right\}}{\sqrt{2\pi\sigma^2 z_1}}
\prod_{i=1}^{N}\frac{{z}_{i+1}^{c_{i}}}{\Gamma(c_i)}{z}_i^{c_{i}-1}\exp\left\{-{z}_{i}{z}_{i+1}\right\}
\right]
d{z}_1 \ldots d{z}_N
\end{equation}

with

\begin{equation}
1-\int_{-k_n}^{k_n}f(x)dx \leq {p_n}^{-(1+u)}
\end{equation}
\begin{equation}
-\log\left(\inf_{x\in [-E_n,E_n]} f(x)\right)=O\left(\log p_n\right)
\end{equation}

where $u>0$ and $k_n=\sqrt{s_n \log p_n/n}/p_n$ then

\begin{equation}
\mathbb{P}_{0}\left(\pi(\parallel\beta-\beta^{n}_0\parallel\geq c_1\sigma_0\epsilon_n|y_n) \geq
e^{-c_2n\epsilon_n^2}\right) \leq e^{-c_3n\epsilon_n^2}
\end{equation}

and

\begin{equation}
\mathbb{P}_{0}\left(\pi(\parallel\beta-\beta^{n}_0 \parallel_1 \geq c_1\sigma^{n}_0 \sqrt{s}\epsilon_n|y_n) \geq
e^{-c_2n\epsilon_n^2}\right) \leq e^{-c_3n\epsilon_n^2}
\end{equation}

for some $c_1,c_2,c_3>0$, where $\sigma_0$ and $\beta^{n}_0$ are the true model parameters and $\mathbb{P}_{0}$ is the probability measure underlying $y_n=X_n \beta_n^0 +\epsilon_n$.

\end{theorem}

Alternatively, for $p_n/n\rightarrow 0$ we can use the main result in \citep{armagan2013posterior}
to show strong posterior consistency by assuming a second set of conditions as

\begin{enumerate}[leftmargin=1cm,label=(B\arabic*):]
\item   $p_n=o(n)$.
\item   Let $\Lambda_{n\,\text{max}}$ and $\Lambda_{n\,\text{min}}$ be the smallest and the largest singular
    values of $X_n$, respectively. Then, $0<\Lambda_\text{min}<\lim \inf_{n\rightarrow\infty}
    \Lambda_{n\,\text{min}}/\sqrt{n} \leq \lim \sup_{n\rightarrow\infty}
    \Lambda_{n\,\text{max}}/\sqrt{n}<\Lambda_\text{max}\leq \infty$.
\item   $\sup_{j=1,\ldots,p_n}|\beta_{nj}^{0}|\leq\infty$.
\item   $s_n=o(n / \log n)$.
\end{enumerate}

The assumptions are sufficiect to show strong posterior consistency for different priors in linear models
 as was shown in \citep{armagan2013posterior} using the following theorem

\begin{theorem}\citep{armagan2013posterior}\label{fff}
under assumptions $(B1)$ and $(B2)$, the posterior of $\beta_n$ under prior $\Pi_n(\beta_n)$ is strongly
consistent if

\begin{equation}
\Pi_n \left( \beta_n:\parallel \beta_n-\beta_n^0 \parallel < \frac{\Delta}{ n^{\rho/2} } \right) > \exp\left\{-dn\right\}
\end{equation}
for all $0<\Delta<\epsilon^2\Lambda_\text{min}^2/(48\Lambda_\text{max}^2)$ and $0<d<\epsilon^2
\Lambda_\text{min}^2/(32\sigma^2)-3\Delta\Lambda_\text{\textbf{max}}^2/(2\sigma^2)$ and some $\rho>0$.

\end{theorem}

We will extended this theorem to show strong posterior consistency for our model using the following theorem

\begin{theorem}\label{sdaa}
Under assumptions $(B1)-(B4)$, the marginal prior given by (\ref{marginal-prior})
is strongly consistent posterior for $\Phi_n=C/(p_n n^\rho \log n)$ with $C,\rho>0$.
\end{theorem}
\begin{proof}
The proof is given in the Appendix.
\end{proof}

\newpage
\section{Simulation Studies \label{simData}}
For our simulation studies, we investigate the prediction accuracy of our proposed model, referred to as $NCG_{10}$ and compare its performance with  Lasso \citep{Tibshirani_1996}, the adaptive Lasso \citep[aLasso,][]{Zou_2006}, SCAD \citep{fan2001variable}, the elastic net \citep[Enet,][]{Zou_and_Hastie_2005}, the minimax concave penalty \citep[MCP,][]{zhang2010nearly}  the horseshoe estimator  \citep{carvalho2010horseshoe} and the Beta prime prior for scale parameters \citep[$NCG_{2}$,][]{bai2018beta}. The Bayesian estimates (the horseshoe, $NCG_{2}$  and our method) are posterior means using 13000 draws of the Gibbs sampler after 2000 draws as burn-in. The data were simulated from the true model (\ref{ahmed_1}). Each generated sample is partitioned into a training set with 20 observations and a testing set with 200 observations. Methods are fitted on the training observations and the mean squared error (MSE) is calculated on the testing set for each method. Then, we calculate the mean of MSE's for the generated samples based on 100 replications.
%%In each simulation study, we generate a training set ($n_t$) and a testing set ($n_p$).
\subsection*{{Simulation 1} (sparse model)}
Here we consider  a sparse  model. We set $\bm\beta=(2, 0, 0, 1, 0, 0, 2, 0, 0, 0)$ and $\sigma^2\in\{1, 9, 25\}$.  The covariates are simulated from the multivariate normal distribution $N(0, \Sigma),$ where $\Sigma$ has one of the following covariance structures:
\begin{itemize}
  \item Case I:  $\Sigma$ is an identity matrix.
  \item Case II: $\Sigma_{ij}=0.5^{|i-j|}$  for all $ 1\leq i\leq j\leq p$.
  \item Case III: $\Sigma_{ij}=0.5$  whenever $i\neq j$, and  $\Sigma_{ii}= 1$ for all $ 1\leq i\leq j\leq p$.
\end{itemize}
% Table 1 presents a summary of the number of correct and wrong zero coefficients based on 100 simulated data sets for three different quantiles.
% latex table generated in R 3.6.2 by xtable 1.8-4 package
% Fri Nov 18 09:01:04 2022

The results are presented in Table 1 for Case I, Table 2 for Case II and Table 3 for Case III. These results show that, in terms of the MSE, the proposed  method (NCG$_{10}$) perform better than the other seven methods in general. It has the smallest
MSE in all cases  (Case I, Case II, and Case III).
Compared with the other seven methods, NCG$_{10}$ produces better false positive rate (FPR) and produces comparable or better false negative rate (FNR) in all cases.

\newpage
\begin{table}[ht]
\caption{Results for Simulation 1 (Case I). All results are averaged over 100 replications and their associated standard deviations (sd) are listed in the parentheses.}
\centering
\begin{tabular}{rrrrrrr}
  \hline
Methods&$\sigma^2$  & MSE (sd) & FPR (sd) & FNR (sd) \\
  \hline
NCG2&1 & 0.3318 (0.2347) & 0.1700  (0.4726) & 0.1200 (0.3266) \\
  NCG10&1 & \textbf{0.2953} (0.2202) & 0.1200  (0.3835) & 0.0500 (0.0589) \\
  MCP &1& 0.2998 (0.2818) & 0.9300  (1.4924) & 0.0200 (0.1407) \\
  SCAD &1& 0.3069 (0.2670) & 1.4700  (1.5920) & 0.0000 (0.0000) \\
  Enet &1& 0.7279 (0.5192) & 0.7800  (1.4466) & 0.0400 (0.1969) \\
  Horseshoe&1& 0.3337 (0.2148) & 0.1400 (0.4025) & 0.1100 (0.3145) \\
  Lasso &1& 0.3904 (0.2575) & 2.8400 (2.0485) & 0.0000 (0.0000) \\
  aLasso &1& 0.3094 (0.3079) & 0.6800  (1.1360) & 0.0400 (0.1969) \\{}\\
  NCG2 &9& 3.3210 (1.6915) & 0.2800 (0.7259) & 1.9000  (1.0200) \\
  NCG10 &9& \textbf{3.3046} (2.0605) & 0.0600 (0.3429) & 2.2700  (0.9519) \\
  MCP &9& 4.6948 (2.7229) & 1.6100 (2.1362) & 0.9500  (1.0577) \\
  SCAD &9& 4.5140 (2.5653) & 1.7800 (1.9047) & 0.7200  (0.9437) \\
  Enet &9& 5.9067 (2.9671) & 0.8600 (1.7351) & 1.6300  (1.2606) \\
  Horseshoe &9& 3.3059 (1.6899) & 0.0700  (0.3555) & 2.2600  (0.9705) \\
  Lasso &9& 3.8683 (2.5425) & 2.4200 (2.2301) & 0.5900  (0.9112) \\
  aLasso &9& 4.8923 (2.7963) & 0.9500 (1.7660) & 1.2000  (1.0347) \\ {}\\
  NCG2 &25& 6.2156 (3.1138) & 0.1400 (0.6199) & 2.5000  (0.7850) \\
  NCG10 &25& \textbf{6.1973} (2.4487) & 0.0600 (0.4221) & 1.6700  (0.5449) \\
  MCP &25& 8.0381 (3.0830) & 1.0400 (1.9945) & 2.0000  (1.1547) \\
  SCAD &25& 7.8252 (3.0894) & 1.1900 (1.8515) & 1.8600  (1.1893) \\
  Enet &25& 8.0952 (2.2614) & 0.4000 (1.3257) & 2.5100  (0.9374) \\
  Horseshoe &25& 6.2033 (2.6232) & 0.0600  (0.4221) & 2.8200  (0.5199) \\
  Lasso &25& 7.6352  (3.2592) & 1.7900  (2.3281) & 1.7200  (1.2314) \\
  aLasso &25& 7.8551  (4.1621) & 0.7100  (1.2251) & 2.0400  (0.9941) \\
   \hline
\end{tabular}
\end{table}

\newpage
\begin{table}[ht]
\caption{Results for Simulation 1 (Case II). }
\centering
\begin{tabular}{rrrrrrr}
  \hline
Methods&$\sigma^2$  & MSE (sd) & FPR (sd) & FNR (sd) \\
  \hline
NCG2 &1& 0.3587  (0.2063)         & 0.2200 (0.6902) & 0.1700 (0.3775) \\
  NCG10 &1& \textbf{0.3039}  (0.1782)        & 0.1000 (0.4381) & 0.0900 (0.1688) \\
  MCP &1&  0.3359  (0.2726)       & 1.0300 (1.7259) & 0.1100 (0.3145) \\
  SCAD &1&  0.3525  (0.2799)      & 1.5500 (1.9456) & 0.0700 (0.2564) \\
  Enet &1&  0.5870  (0.3823)      & 1.8200 (1.3437) & 0.0000 (0.0000) \\
  Horseshoe &1&  0.3340  (0.1774) & 0.1100  (0.3994) & 0.2600 (0.4408) \\
  Lasso &1&  0.3567  (0.2267)     & 3.2400  (1.7759) & 0.0000 (0.0000) \\
  aLasso &1&  0.3183 (0.2459)    & 0.9400  (1.4759) & 0.0800 (0.2727) \\{}\\
  NCG2 &3& 2.7638 (1.3834) & 0.1800  (0.5390) & 2.0600  (0.8507) \\
  NCG10 &3& \textbf{2.7144} (1.4640) & 0.0400 (0.2429) & 0.5400  (0.5759) \\
  MCP &3& 4.6676 (2.5626) & 1.6800  (2.1124) & 0.9300  (0.7688) \\
  SCAD &3& 4.4064 (2.0405) & 2.1300 (1.9209) & 0.6600  (0.6700) \\
  Enet &3& 5.3799 (3.7182) & 1.4300  (1.4443) & 0.8300  (0.8294) \\
  Horseshoe &3& 2.7479 (1.3851) & 0.0900 (0.3208) & 2.3300  (0.6825) \\
  Lasso &3& 2.9178 (1.6844) & 2.5600 (1.8053) & 0.3900  (0.5667) \\
  aLasso &3& 3.9332 (2.6644) & 1.2100  (1.4859) & 0.9300 (0.7420) \\{}\\
  NCG2 &5& 5.5666  (3.4173) & 0.1800  (0.5001) & 2.6700  (0.6365) \\
  NCG10 &5& \textbf{5.2849}  (2.9652) & 0.0500  (0.2190) & 1.8100 (0.4191) \\
  MCP &5& 10.4075  (5.8529) & 1.5200  (1.9974) & 1.7400 (0.8833) \\
  SCAD &5& 8.3466  (5.1968) & 1.9000  (1.5859) & 1.4600  (0.8339) \\
  Enet &5& 10.7077  (5.7116) & 0.9500 (1.2663) & 2.0000 (1.0541) \\
  Horseshoe &5& 5.5427  (3.0637) & 0.0600 (0.2778) & 2.7800 (0.4623) \\
  Lasso &5& 6.7542  (4.1350) & 2.4500 (1.7887) & 1.2000  (0.8409) \\
  aLasso &5& 9.0388 (5.1451) & 1.3900 (1.6569) & 1.7400  (0.8483) \\
   \hline
\end{tabular}
\end{table}

\newpage
\begin{table}[ht]
\caption{Results for Simulation 1 (Case III).}
\centering
\begin{tabular}{rrrrrrr}
  \hline
Methods&$\sigma^2$  & MSE (sd) & FPR (sd) & FNR (sd) \\
  \hline
NCG2 &1& 0.3684  (0.1756) & 0.2600 (0.7865) & 0.2400 (0.4292) \\
  NCG10 &1& \textbf{0.3211}  (0.1762) & 0.1400 (0.6516) & 0.0500 (0.4648) \\
  MCP &1& 0.3240  (0.2381) & 1.2000 (1.9488) & 0.0200 (0.1407) \\
  SCAD &1& 0.3318  (0.2343) & 1.7700 (2.0144) & 0.0100 (0.1000) \\
  Enet &1& 0.7019  (0.3880) & 1.6900 (1.5485) & 0.0000 (0.0000) \\
  Horseshoe &1& 0.3499  (0.1626) & 0.1500 (0.6256) & 0.2600 (0.4408) \\
  Lasso &1& 0.4167  (0.1967) & 3.1300  (1.9779) & 0.0000  (0.0000) \\
  aLasso &1& 0.3237 (0.2402) & 1.0300 (1.5599) & 0.0200 (0.1407) \\{}\\
  NCG2 &3& 3.0482 (1.4998) & 0.3100  (0.7344) & 1.9900 (0.9374) \\
  NCG10 &3& \textbf{3.0153} (1.7401) & 0.0800  (0.3387) & 0.9100 (0.8420) \\
  MCP &3& 4.8675 (2.9906) & 1.8000  (2.0792) & 0.9700 (1.0392) \\
  SCAD &3& 4.7276  (2.7093) & 2.3100  (2.0631) & 0.7800 (0.9383) \\
  Enet &3& 5.9578 (3.6717) & 1.1900  (1.6434) & 1.5800 (1.2806) \\
  Horseshoe &3& 2.9680  (1.4774) & 0.0700  (0.2932) & 2.3300 (0.8294) \\
  Lasso &3& 3.5033  (2.2545) & 3.0100  (1.9669) & 0.5000 (0.8103) \\
  aLasso &3& 4.7234  (3.0103) & 1.4100 (1.7984) & 1.1700 (1.0156) \\{}\\
NCG2 &5& \textbf{6.2475 } (4.1274) & 0.1900  (0.6771) & 2.6600 (0.6995) \\
  NCG10 &5& 6.5306 (3.5098) & 0.0700  (0.4324) & 1.8800 (0.4330) \\
  MCP &5& 10.1808  (6.1006) & 1.3500  (2.0762) & 1.8200 (1.1315) \\
  SCAD &5& 9.5614 (6.1075) & 1.7700  (2.1075) & 1.6700 (1.1197) \\
  Enet &5& 9.1984  (2.5258) & 0.5200  (1.3141) & 2.5000 (0.9156) \\
  Horseshoe &5& 6.0867 (3.7116) & 0.0600  (0.3429) & 2.8200 (0.4579) \\
  Lasso &5& 8.2867 (4.9088) & 2.2100  (2.3063) & 1.3900 (1.1712) \\
  aLasso &5& 9.1999 (5.0330) & 1.0500  (1.6291) & 1.9700 (1.0294) \\
   \hline
\end{tabular}
\end{table}

\newpage
\subsection*{{Simulation 2} (very sparse model)}
In this simulation study, we investigate the performance of our proposed method  in a very sparse model.
We set $\bm\beta=(1, 0, 0, 0, 0, 0, 0, 0, 0, 0)$ and $\sigma^2\in\{1, 9, 25\}$.  The covariates are simulated as the same as in Simulation 1.

The results are presented in Table 4 for Case I, Table 5 for Case II and Table 6 for Case III. Again, the proposed  method (NCG$_{10}$) perform better than the other seven methods in general. It has the smallest
MSE in most cases. Compared with the other seven methods, NCG$_{10}$ produces  comparable or better  false positive rate   and   false negative rate in all cases.

\newpage
\begin{table}[ht]
\caption{Results for Simulation 2 (Case I).}
\centering
\begin{tabular}{rrrrrrr}
  \hline
Methods&$\sigma^2$  & MSE (sd) & FPR (sd) & FNR (sd) \\
  \hline
NCG2 &1& 0.2285  (0.1914) & 0.2200 (0.5427) & 0.1400 (0.3487) \\
  NCG10 &1& \textbf{0.1856}  (0.2151) & 0.0800 (0.3075) & 0.0900 (0.1230) \\
  MCP &1& 0.2224  (0.2719) & 1.3500 (2.0019) & 0.0400 (0.1969) \\
  SCAD &1& 0.2415  (0.2952) & 2.1300 (2.2095) & 0.0500 (0.2190) \\
  Enet &1& 0.4422  (0.3437) & 0.6300 (1.4331) & 0.2200 (0.4163) \\
  Horseshoe &1& 0.2136  (0.1840) & 0.0800 (0.3075) & 0.2300  (0.4230) \\
  Lasso &1& 0.2624 (0.2474) & 2.5600 (2.6221) & 0.0500  (0.2190) \\
  aLasso &1& 0.2497  (0.2958) & 0.6800 (1.4831) & 0.0700  (0.2564) \\{}\\
  NCG2 &3& 1.4275 (0.9386) & 0.1600 (0.5635) & 0.9100  (0.2876) \\
  NCG10 &3& 1.0412 (0.6628) & 0.0400 (0.1969) & 0.8600 (0.1969) \\
  MCP &3& 1.5401 (1.5609) & 0.9100 (1.5960) & 0.6900   (0.4648) \\
  SCAD &3& 1.5115 (1.3905) & 1.3500 (1.9300) & 0.6600  (0.4761) \\
  Enet &3& \textbf{1.0177} (0.2889) & 0.3100  (1.1866) & 0.9100  (0.2876) \\
  Horseshoe &3& 1.2825 (0.8173) & 0.0400 (0.1969) & 0.9600  (0.1969) \\
  Lasso &3& 1.4042 (1.0672) & 1.6600 (2.1613) & 0.5600  (0.4989) \\
  aLasso &3& 1.3302 (1.2648) & 0.5800 (1.0841) & 0.7600 (0.4292) \\{}\\
  NCG2 &5& 3.2319  (3.3268) & 0.2000  (0.6513) & 0.9200 (0.2727) \\
  NCG10 &5& \textbf{1.8661} (2.2425) & 0.0400  (0.1969) & 0.8700 (0.1714) \\
  MCP &5& 3.3244 (4.8079) & 0.8000  (1.6937) & 0.8400 (0.3685) \\
  SCAD &5& 3.2317  (5.1025) & 1.1200  (1.9137) & 0.8000 (0.4020) \\
  Enet &5& 1.9996 (1.8991) & 0.2400  (1.0359) & 0.9300 (0.2564) \\
  Horseshoe &5& 2.7378  (2.6782) & 0.0300  (0.1714) & 0.9700 (0.1714) \\
  Lasso &5& 2.6915  (3.3954) & 1.3800  (2.1686) & 0.7600 (0.4292) \\
  aLasso &5& 2.5437  (3.9197) & 0.5900  (1.1555) & 0.8500 (0.3589) \\
   \hline
\end{tabular}
\end{table}

\newpage
\begin{table}[ht]
\caption{Results for Simulation 2 (Case II).}
\centering
\begin{tabular}{rrrrrrr}
  \hline
Methods&$\sigma^2$  & MSE (sd) & FPR (sd) & FNR (sd) \\
  \hline
NCG2 &1& 0.2032  (0.1427) & 0.1700 (0.6039) & 0.3900 (0.4902) \\
  NCG10 &1& \textbf{0.1768} (0.1637) & 0.0100 (0.1000) & 0.0900 (0.5024) \\
  MCP &1& 0.2415 (0.2768) & 0.8500 (1.4240) & 0.0500 (0.2190) \\
  SCAD &1& 0.2236  (0.2404) & 1.6500 (1.7313) & 0.0100 (0.1000) \\
  Enet &1& 0.4617  (0.3507) & 0.6400 (1.1238) & 0.2300 (0.4230) \\
  Horseshoe &1& 0.1927 (0.1345) & 0.0000 (0.0000) & 0.4400 (0.4989) \\
  Lasso &1& 0.2006 (0.1698) & 2.3800 (2.1498) & 0.0100 (0.1000) \\
  aLasso &1& 0.2373  (0.3005) & 0.6600 (1.3198) & 0.0900 (0.2876) \\{}\\
  NCG2 &3& 1.1350 (1.1208) & 0.1500  (0.5389) & 0.9600  (0.1969) \\
  NCG10 &3& \textbf{0.7752} (0.4526) & 0.0000  (0.0000) & 0.8200  (0.0000) \\
  MCP &3& 1.7061 (1.8549) & 0.7700  (1.5233) & 0.8400  (0.3685) \\
  SCAD &3& 1.6668 (1.9736) & 1.0500  (1.6840) & 0.8100  (0.3943) \\
  Enet &3& 1.0071 (0.3125) & 0.2200  (0.8113) & 0.9300  (0.2564) \\
  Horseshoe &3& 1.0181 (0.7426) & 0.0000  (0.0000) & 1.0000  (0.0000) \\
  Lasso &3& 1.3780 (1.3094) & 1.4600  (2.2627) & 0.7200  (0.4513) \\
  aLasso &3& 1.1627 (1.1377) & 0.4800 (1.1054) & 0.8100  (0.3943) \\{}\\
CG2 &5& 2.4748  (3.1561) & 0.0800  (0.4188) & 0.9600  (0.1969) \\
  NCG10 &5& 1.5892  (2.6613) & 0.0400  (0.2429) & 0.8500  (0.1000) \\
  MCP &5& 3.0585  (5.1631) & 0.5800  (1.3041) & 0.8500  (0.3589) \\
  SCAD &5& 3.4144  (5.5803) & 0.9600  (1.9588) & 0.8200  (0.3861) \\
  Enet &5& \textbf{1.2142}  (1.0883) & 0.1600  (0.7208) & 0.9100  (0.2876) \\
  Horseshoe &5& 2.2256  (2.8491) & 0.0400  (0.2429) & 0.9900  (0.1000) \\
  Lasso &5& 2.5749  (3.4556) & 1.3400  (2.0803) & 0.7600  (0.4292) \\
  aLasso &5& 1.9558  (3.1287) & 0.4500  (0.7961) & 0.8300  (0.3775) \\

   \hline
\end{tabular}
\end{table}

\newpage
\begin{table}[ht]
\caption{Results for Simulation 2 (Case III).}
\centering
\begin{tabular}{rrrrrrr}
  \hline
Methods&$\sigma^2$  & MSE (sd) & FPR (sd) & FNR (sd) \\
  \hline
NCG2 &1& 0.2403   (0.1694) & 0.2400  (0.5341) & 0.1400 (0.3487) \\
  NCG10 &1& \textbf{0.1920} (0.1897) & 0.0900 (0.2876) & 0.0900 (0.4560) \\
  MCP &1& 0.2250 (0.2446) & 1.4600  (2.0222) & 0.0200 (0.1407) \\
  SCAD &1& 0.2361 (0.2463) & 2.1300 (2.0580) & 0.0100 (0.1000) \\
  Enet &1& 0.4537 (0.2998) & 0.7000 (1.2831) & 0.1700 (0.3775) \\
  Horseshoe &1& 0.2271 (0.1690) & 0.1100  (0.3451) & 0.2500 (0.4352) \\
  Lasso &1& 0.2503 (0.1888) & 2.6600 (2.5593) & 0.0000 (0.0000) \\
  aLasso &1& 0.2353 (0.2598) & 0.8000 (1.2144) & 0.0500 (0.2190) \\{}\\
NCG2 &3& 1.2270  (1.0247) & 0.1400 (0.5689) & 0.9200 (0.2727) \\
  NCG10 &3& \textbf{0.9459}  (0.5935) & 0.0200 (0.1407) & 0.7600 (0.1969) \\
  MCP &3& 1.6726  (1.8068) & 0.9300 (1.8327) & 0.6800 (0.4688) \\
  SCAD &3& 1.6060  (1.8963) & 1.1300 (1.9051) & 0.6600 (0.4761) \\
  Enet &3& 1.0271  (0.5398) & 0.2400 (0.8776) & 0.8800 (0.3266) \\
  Horseshoe &3& 1.1431  (0.8349) & 0.0200  (0.1407) & 0.9600 (0.1969) \\
  Lasso &3& 1.4559  (1.4759) & 1.5600 (2.4176) & 0.5800 (0.4960) \\
  aLasso &3& 1.3460  (1.1822) & 0.5600 (1.1748) & 0.7400 (0.4408) \\{}\\
NCG2 &5& 2.6455 (2.8117) & 0.0800      (0.4645) & 0.9800  (0.1407) \\
  NCG10 &5& 1.5017 (1.9620) & 0.0400  (0.3153) & 0.7100  (0.0000) \\
  MCP &5& 2.9577 (4.4169) & 0.8300  (1.7925) & 0.7800  (0.4163) \\
  SCAD &5& 2.4440 (3.7348) & 0.9000  (1.5472) & 0.8300  (0.3775) \\
  Enet &5& \textbf{1.1314} (1.1063) & 0.1400  (0.8167) & 0.9600 (0.1969) \\
  Horseshoe &5& 2.2756 (2.4153) & 0.0400  (0.3153) & 1.0000 (0.0000) \\
  Lasso &5& 2.5601 (3.3173) & 1.4700  (2.2539) & 0.7400  (0.4408) \\
  aLasso &5& 1.7135  (2.5835) & 0.5700  (1.1304) & 0.8300  (0.3775) \\
   \hline
\end{tabular}
\end{table}

\newpage
\subsection*{{Simulation 3} (Difficult Example)}
Here,  we consider a difficult example which is provided in the adaptive lasso paper \citep{zou2006adaptive}. This example requires setting $\bm{\beta}=(5.6, 5.6, 5.6, 0)'$, the correlation coefficient between $x_i$ and $x_j$ is equal to -0.39 for $i<j<4$  and the correlation coefficient between $x_i$ and $x_4$ is equal to 0.23 for $i<4$.
The results are summarized  in Table 7. The results show that  the proposed  method (NCG$_{10}$) perform better than (NCG$_{2}$) for $n= 20, 50, 100, 150, 200$ and 250. It has the smallest MSE  better  false positive rate in all situations. The convergence of the Gibbs sampler  was evaluated  by trace plots of the simulated draws. Figures \ref{ncg2} and \ref{ncg10} displays the trace plots for NCG$_{2}$ and NCG$_{10}$, respectively. It can be seen that the trace plots of   NCG$_{10}$ establish
good mixing property of the proposed Gibbs sampler. Additionally, The histograms of NCG$_{10}$ based on the posterior draws reveal that the
conditional distributions for the regression coefficients are the desired stationary distributions.

\begin{table}[ht]
\caption{Results for Simulation 3.}
\centering
\begin{tabular}{rrrrrrr}
  \hline
Methods&$n$  & MSE (sd) & FPR (sd) & FNR (sd) \\
  \hline
NCG2 & 20 & 0.2059 (0.1464) & 0.0500 (0.2190) & 0.0000 (0.0000) \\
  NCG10 & 20  & \textbf{0.1915} (0.1354) & 0.0300 (0.1714) & 0.0000 (0.0000) \\ {} \\
NCG2 & 50  & 0.0923 (0.0659) & 0.0600 (0.2387) & 0.0000 (0.0000) \\
  NCG10 & 50  & \textbf{0.0881} (0.0638) & 0.0300 (0.1714) & 0.0000 (0.0000) \\ {}\\
NCG2 & 100  & 0.0368 (0.0232) & 0.0500 (0.2190) & 0.0000 (0.0000) \\
  NCG10 & 100  & \textbf{0.0347} (0.0225) & 0.0400 (0.1969) & 0.0000 (0.0000) \\ {} \\
NCG2 & 150  & 0.0275 (0.0191) & 0.0600 (0.2387) & 0.0000 (0.0000) \\
  NCG10 & 150  & \textbf{0.0261} (0.0186) & 0.0500 (0.2190) & 0.0000 (0.0000) \\ {} \\
NCG2 & 200  & 0.0204 (0.0123) & 0.0400 (0.1969) & 0.0000 (0.0000) \\
  NCG10 & 200  & \textbf{0.0195} (0.0122) & 0.0300 (0.1714) & 0.0000 (0.0000) \\ {} \\
NCG2 & 250  & 0.0166 (0.0136) & 0.0800 (0.2727) & 0.0000 (0.0000) \\
  NCG10 & 250  & \textbf{0.0160} (0.0134) & 0.0500 (0.2190) & 0.0000 (0.0000) \\
   \hline
\end{tabular}
\end{table}

\begin{figure}[h]
\begin{center}
    \includegraphics[width=120mm]{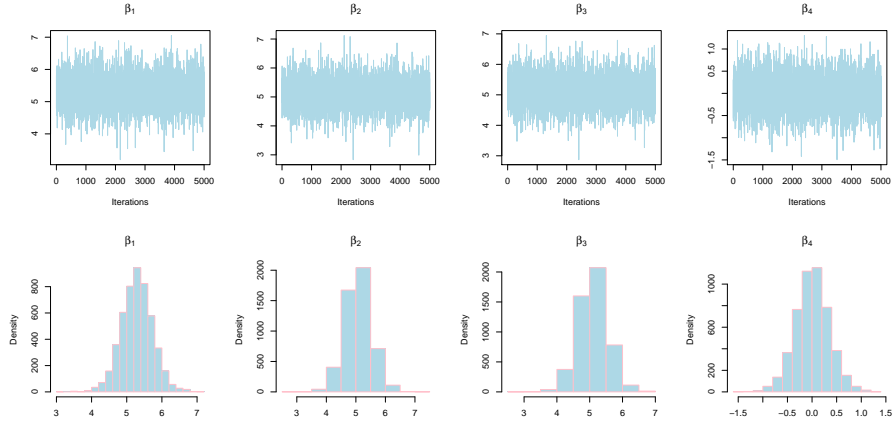}
    \caption{Trace plots and histograms of Simulation 3 using NCG$_{2}$.\label{ncg2}}
    \label{figure for different N and a}
\end{center}
\end{figure}

\begin{figure}[h]
\begin{center}
    \includegraphics[width=120mm]{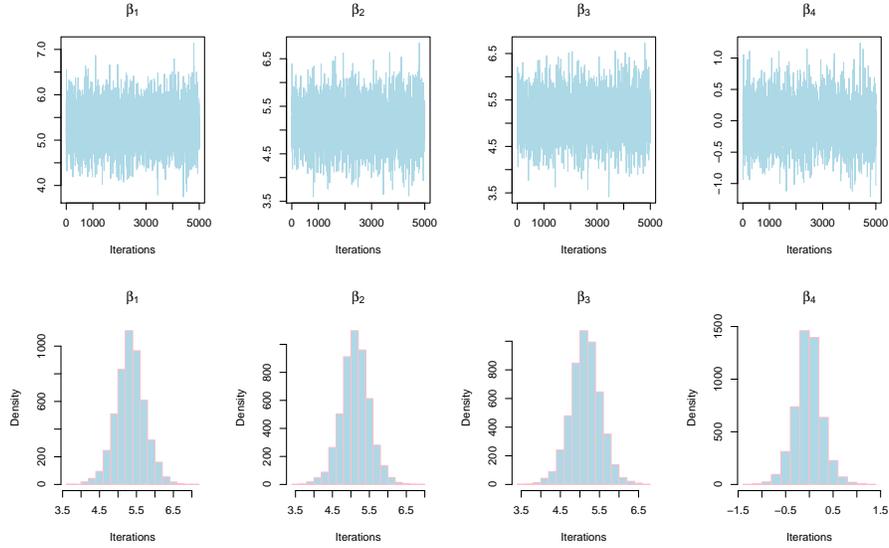}
    \caption{Trace plots and histograms of Simulation 3 using NCG$_{10}$. \label{ncg10}}
    \label{figure for different N and a}
\end{center}
\end{figure}

\newpage
\section{A real data example \label{realData}}
In this section, we compare the performance of the eight methods in Section \ref{simData}, NCG$_2$, NCG$_{10}$,
 MCP, SCAD, Enet, Horseshoe, Lasso, aLasso, on the prostate cancer  data \citep{stamey1989prostate}. This data set was used  for illustration in previous studies  \citep[see,][]{zou2006adaptive,Park_and_Casella,mallick2014new},  where the dependent  variable  is the logarithm of prostate-specific antigen and the number of covariates are eight clinical measures donated by lcavol = the logarithm of cancer volume, lweight = the logarithm of prostate weight, age, lbph = the logarithm of the amount of benign prostatic hyperplasia, svi = seminal vesicle invasion, lcp = the logarithm of capsular penetration, gleason = the Gleason score, and pgg45 = the percentage Gleason score 4 or 5. We add  twelve more predictor  noise standard normal random variables, $x_{9}, \cdots, x_{20}$.
   To analyze the data, we follow the methods in \citep{Zou_and_Hastie_2005} and \citep{mallick2014new} by dividing it into a training set with 67 observations  and a testing set with 30 observations. The results are summarized in Table \ref{realdataTable}. We can see clearly, that the proposed method performs better than the  other methods.

% latex table generated in R 3.6.2 by xtable 1.8-4 package
% Sat Dec 10 22:59:48 2022
\begin{table}[ht]
\caption{Results for the real data.  \label{realdataTable} }
\centering
\begin{tabular}{rrrrrrr}
  \hline
%Methods  & MSE (sd) & FPR (sd) & FNR (sd) \\
 % \hline
%NCG2 & 18.5125 (0.3134) & 0.1200 (0.4798) & 0.4400 (0.5406) \\
 % NCG10 & 18.4303 (0.2142) & 0.0600 (0.3136) & 0.8400 (0.4219) \\
  %MCP & 19.4068 (0.6941) & 0.1800 (0.3881) & 0.8400 (0.3703) \\
  %SCAD & 19.4802 (0.6962) & 0.9400 (0.2399) & 0.8800 (0.3283) \\
  %Enet & 20.6716 (0.7222) & 0.3400 (0.4785) & 1.0000 (0.0000) \\
  %Horseshoe & 18.7082 (0.2732) & 0.0600 (0.3136) & 6.7400 (0.4870) \\
  %Lasso & 18.7504 (0.4828) & 1.0000 (0.0000) & 0.9000 (0.3030) \\
  %aLasso & 19.2235 (0.8163) & 0.4600 (0.5035) & 0.9800 (0.1414) \\
  %Methods  & MSE (sd) & FPR (sd) & FNR (sd) \\
  %\hline
Methods  & MSE (sd) & FPR (sd)  \\
NCG2 & 18.5125 (0.3134) & 0.1200 (0.4798)  \\
  NCG10 & 18.4303 (0.2142) & 0.0600 (0.3136) \\
  MCP & 19.4068 (0.6941) & 0.1800 (0.3881)  \\
  SCAD & 19.4802 (0.6962) & 0.9400 (0.2399)  \\
  Enet & 20.6716 (0.7222) & 0.3400 (0.4785)  \\
  Horseshoe & 18.7082 (0.2732) & 0.0600 (0.3136) \\
  Lasso & 18.7504 (0.4828) & 1.0000 (0.0000)  \\
  aLasso & 19.2235 (0.8163) & 0.4600 (0.5035)  \\
   \hline
\end{tabular}
\end{table}

\newpage
\section{Concluding Remarks}

In this paper, we have studied the properties of the normal prior with a compound gamma scale mixture. Then, we  proposed a
Gibbs sampler and a Variational Bayes method to study the model predictions. Additionally, we incorporated the latter two methods to an EM method to evaluate the corresponding model hyperparameters. Furthermore, we studied posterior consistency of our model using two different sets of conditions for $p_n=o(n)$ and $p_n\gg n$. Finally, we illustrated our model using simulated and real data examples.

\section{Appendix}

\begin{proof}[Proof of Proposition \ref{equivalence-proposition}.]

The proof can be trivially seen using a simple change of variable of integration through $z_i\rightarrow x^{(-1)^{i+1}} z_{}^{(-1)^i}$

\begin{align}
\begin{split}
\pi(x) &= \int_0^\infty\ldots\int_0^\infty
\exp\{-z_2\}z_2^{c_1-1}
\left[
\prod_{i=2}^{M-1}\exp\left\{-\left(\frac{z_{i+1}}{z_i}\right)^{(-1)^{i+1}}\right\}
\left(\frac{z_{i+1}}{z_i}\right)^{(-1)^{i+1} c_i-1}\frac{1}{z_i}
\right]
\\&
\qquad\qquad\qquad\qquad\qquad
\times
\exp\left\{-\phi\left(\frac{z_1}{z_N}\right)^{(-1)^{N+1}}\right\}
\left(\frac{z_1}{z_{N}}\right)^{(-1)^{N+1} c_{N}-1}\frac{\phi ^{c_N}}{z_N}
d{z}_2 \ldots d{z}_N
\\&= \int_0^\infty\ldots\int_0^\infty
\mathcal{G}(z_2,c_1,1)\frac{\mathcal{IG}(z_3/z_3,c_2,1)}{z_2}\frac{\mathcal{G}(z_3/z_4,c_3,1)}{z_3}
\\&
\qquad\qquad\qquad\qquad\qquad\qquad\qquad\qquad\qquad\qquad\qquad
\ldots \times\frac{\mathcal{AG}(z_1/z_N,N,c_N,\phi)}{z_N}
d{z}_2 \ldots d{z}_N
\end{split}
\end{align}

\end{proof}

\begin{proof}[Proof of Theorem \ref{pythagorean}.]

\begin{lemma}
If ${c_1}_n\rightarrow 0$ and $c_2\in (1,\infty)$ as $n\rightarrow\infty$, then
$\frac{\Gamma({c_1}_n+c_2)}{\Gamma({c_1}_n)\Gamma(c_2)}\asymp {c_1}_n$.
\end{lemma}
\begin{proof}
The details can be found in \citep{bai2018beta}.
\end{proof}

By the symmetry of the probability for a single density function given by (\ref{marginal-prior})

\begin{align}
\begin{split}
1-\int_{-k_n}^{k_n}f(x)dx &= 2\int_{k_n}^{\infty}f(x)dx\\
&\leq 2\int_0^\infty\ldots\int_0^\infty
\left[
\frac{\exp\left\{-k_n^2/{2 z_1}\right\}}{\sqrt{2\pi z_1}}
\prod_{i=1}^{N}\frac{{z}_{i+1}^{c_{i}}}{\Gamma(c_i)}{z}_i^{c_{i}-1}\exp\left\{-{z}_{i}{z}_{i+1}\right\}
\right]
d{z}_1 \ldots d{z}_N\\
&\leq 2 c_{1_n}  \int_0^\infty \ldots\int_0^\infty\int_0^\infty
\bigg[
\frac{\exp\left\{-k_n^2/{2 z_1}\right\}}{\sqrt{2\pi z_1}}
\frac{{z}_1^{c_1-1}{z}_3^{c_2}}{({z}_1+{z}_3)^{c_1+1}}\\
&\qquad\qquad\qquad\qquad\qquad\qquad\times
\prod_{i=3}^{N}\frac{{z}_{i+1}^{c_{i}}}{\Gamma(c_i)}{z}_i^{c_{i}-1}\exp\left\{-{z}_{i}{z}_{i+1}\right\}
\bigg]
d{z}_1d{z}_3\ldots d{z}_N\\
&\leq 2 c_{1_n}  \int_0^\infty \ldots\int_0^\infty\int_0^\infty
\bigg[
\frac{\exp\left\{-k_n^2 u/2\right\}}{({z}_3u+1)^{c_1+1}}\\
&\qquad\qquad\qquad\qquad\qquad\qquad\times
{z}_3^{c_2}\prod_{i=3}^{N}\frac{{z}_{i+1}^{c_{i}}}{\Gamma(c_i)}{z}_i^{c_{i}-1}\exp\left\{-{z}_{i}{z}_{i+1}\right\}
\bigg]
du d{z}_3\ldots d{z}_N\\
&\leq \frac{8 c_{1_n}}{k_n^2} \int_0^\infty \ldots\int_0^\infty
\bigg[
{z}_3^{c_2}\prod_{i=3}^{N}\frac{{z}_{i+1}^{c_{i}}}{\Gamma(c_i)}{z}_i^{c_{i}-1}\exp\left\{-{z}_{i}{z}_{i+1}\right\}
\bigg]
d{z}_3\ldots d{z}_N\\
&= \frac{8c_{1_n}}{k_n^2}\frac{\Gamma[c_2+c_3]}{\Gamma (c_3)}
\int_0^\infty\ldots\int_0^\infty
\bigg[
z_4^{-c_{2}}\prod_{i=4}^{N}\frac{{z}_{i+1}^{c_{i}}}{\Gamma(c_i)}{z}_i^{c_{i}-1}\exp\left\{-{z}_{i}{z}_{i+1}\right\}
\bigg]
d{z}_4\ldots d{z}_N\\
\end{split}
\end{align}

where we have used a change of variable $u=1/z_1$ in the third inequality. By repeating the last step, we get

\begin{align}
\begin{split}
1-\int_{-k_n}^{k_n}f(x)dx \leq \frac{8c_{1_n}\phi^{(-1)^{N}c_2}}{k_n^2 } \prod_{i=3}^{N}\frac{\Gamma[(-1)^{i+1}c_2+c_i]}{\Gamma [c_i]}\\
\end{split}
\end{align}

where we have used the concentration inequality, $\Pr(|X|>k_n)\leq 2 e^{-k_n^2 /{2z_1}}$, and the assumption $b\in(1,\infty)$. For the second condition we have

\begin{align}
\begin{split}
\inf_{x\in [-E_n,E_n]} f(x)
&\gtrsim c_{1_n}  \int_0^\infty \ldots\int_0^\infty\int_0^\infty
\bigg[
\frac{{z}_1^{c_1-3/2}{z}_3^{c_2}\exp\left\{-{x^2}/{\left(2z_1\right)}\right\}}{({z}_1+{z}_3)^{c_1+c_2}}\\
&\qquad\qquad\qquad\qquad\times
\prod_{i=3}^{N}\frac{{z}_{i+1}^{c_{i}}}{\Gamma(c_i)}{z}_i^{c_{i}-1}\exp\left\{-{z}_{i}{z}_{i+1}\right\}
\bigg]
d{z}_1d{z}_3\ldots d{z}_N\\
&\geq
c_{1_n} (x^2)^{c_1-\frac{1}{2}}\int_0^\infty \ldots\int_0^\infty\int_1^2
\biggr[\prod_{i=3}^{N}\frac{{z}_{i+1}^{c_{i}}}{\Gamma(c_i)}{z}_i^{c_{i}-1}\exp\left\{-{z}_{i}{z}_{i+1}\right\}\\
&\times
\exp\left(-\zeta \right)\left(\frac{\zeta}{x^2+2\zeta {z}_3}\right)^{c_2-1/2}\left(\frac{1}{x^2+2\zeta {z}_3}\right)^{c_1+1/2}
\biggr]
d\zeta d{z}_3\ldots d{z}_N\\
&\gtrsim
c_{1_n} (x^2)^{c_1-\frac{1}{2}}\int_0^\infty \ldots\int_0^\infty\int_1^2
\biggr[
z_3^{c_2}\prod_{i=3}^{N}\frac{{z}_{i+1}^{c_{i}}}{\Gamma(c_i)}{z}_i^{c_{i}-1}\exp\left\{-{z}_{i}{z}_{i+1}\right\}\\
&\qquad\qquad\times
\left(\frac{1}{x^2+2{z}_3}\right)^{c_2-1/2}\left(\frac{1}{x^2+4{z}_3}\right)^{c_1+1/2}
\biggr]
d{z}_3 d{z}_4\ldots d{z}_N\\
&\gtrsim
\frac{c_{1_n}}{\left(x^2\right)^{c_2+1/2}}\int_0^\infty \ldots\int_0^\infty
\biggr[
z_4^{c_3}\prod_{i=4}^{N}\frac{{z}_{i+1}^{c_{i}}}{\Gamma(c_i)}{z}_i^{c_{i}-1}\exp\left\{-{z}_{i}{z}_{i+1}\right\}
\biggr]
d{z}_4\ldots d{z}_N\\
\end{split}
\end{align}

using ${z}_i^{c_{i}+c_{i+1}-1} \geq 1, \forall z_i \in [1,2]$ and where we have used $\zeta=\frac{x^2}{2{z}_1}$ in the second inequality. Using assumptions $\log E_n = O(\log p_n)$ and $c_{1_n} \lesssim k_n^2 p_n^{-(1+u)}$ for $u>0$ then

\begin{equation}
-\log\left(\inf_{x\in [-E_n,E_n]}f(x)\right)\lesssim \log p_n
\end{equation}

\end{proof}

\begin{proof}[Proof of Theorem \ref{sdaa}]

Using equation $(A3)$ in \citep{armagan2013posterior}

\begin{align}
\begin{split}
\Pi_n &\left( \beta_n:\parallel \beta_n-\beta_n^0 \parallel < \frac{\Delta}{ n^{\rho/2} } \right) \geq
\left\{ 1-\frac{p_n n^\rho \mathbb{E}(\beta^2_{n_j})}{\Delta^2} \right\}
 \\& \prod_{j\in \Theta_n} \left\{ \Pi_n\left( \beta_{n_j} : |\beta_{n_j}-\beta^0_{n_j}| <\frac{\Delta}{\sqrt{p_n} n^{\rho/2}} \right) \right\} \\
\end{split}
\end{align}

Using $\mathbb{E}(\beta^2)=\frac{\sigma^2 c_1 c_3 c_5 \ldots c_N \phi^{(-1)^N}}{(c_2-1)(c_4-1)\ldots (c_{N-1}-1)}$ and the second part of the proof of Theorem \ref{pythagorean}, then

\begin{align}
\begin{split}
\Pi_n &\left( \beta_n:\parallel \beta_n-\beta_n^0 \parallel < \frac{\Delta}{ n^{\rho/2} } \right) \geq
 \left\{ 1-\frac{p_n n^\rho }{\Delta^2} \frac{\sigma^2 c_1 c_3 c_5 \ldots c_N \phi^{(-1)^N}}{(c_2-1)(c_4-1)\ldots (c_{N-1}-1)}   \right\} \\
&\times
\left[
\frac{2\Delta c_{1_n}}{\sqrt{p_n}n^{\rho/2}}
\left\{   \frac{\sup_{j\in\mathcal{A}_n} (\beta^{0}_{nj})^2}{\phi_n} + \frac{\Delta^2}{p_n n^\rho \phi_n}   \right\}^{-c_2-1/2}
 \right]^{s_n}\\
\end{split}
\end{align}
\end{proof}

letting $\phi_n=C/(p_n n^\rho \log n)$ for some $C>0$ and taking the negative log of both side, we can easily see that $\Pi_n \left( \beta_n:\parallel \beta_n-\beta_n^0 \parallel < \frac{\Delta}{ n^{\rho/2} } \right)<dn, \forall d>0$ which completes the prove.

\newpage
\bibliographystyle{chicago}%IEEEtran
\renewcommand{\baselinestretch}{1}
\normalsize
\clearpage%
\phantomsection%
\addcontentsline{toc}{chapter}{\numberline{}{Bibliography}}%
\bibliography{Ref}  %\bibliography{ref}

\end{document}